\documentclass[10pt, doublecolumn]{IEEEtran}
\usepackage{epsfig,latexsym}
\usepackage{float}
\usepackage{indentfirst}
\usepackage{amsmath}
\usepackage{bm}
\usepackage{amssymb}
\usepackage{times}

\usepackage{algorithm}
\usepackage[noend]{algpseudocode}
\usepackage{subfigure}
\usepackage{psfrag}
\usepackage{hyperref}
\usepackage{cite}
\usepackage{lastpage}
\usepackage{fancyhdr}
\usepackage{color}
 \usepackage{amsthm}
\usepackage{bigints}
\newtheorem{Proposition}{Proposition}

\newtheorem{proposition}[Proposition]{Proposition}

\newcounter{problem}
\newcounter{save@equation}
\newcounter{save@problem}
\makeatletter
\newenvironment{problem}
{\setcounter{problem}{\value{save@problem}}%
  \setcounter{save@equation}{\value{equation}}%
  \let\c@equation\c@problem
  \subequations
}
{\endsubequations
  \setcounter{save@problem}{\value{equation}}%
  \setcounter{equation}{\value{save@equation}}%
}

\begin{document}
\title{ \huge{ NOMA Beamforming in SDMA Networks: \\Riding on Existing   Beams or Forming   New Ones?    }}

\author{ Zhiguo Ding, \IEEEmembership{Fellow, IEEE}    \thanks{ 
  
\vspace{-2em}

    Z. Ding
 is    with the School of
Electrical and Electronic Engineering, the University of Manchester, Manchester, UK (email: \href{mailto:zhiguo.ding@manchester.ac.uk}{zhiguo.ding@manchester.ac.uk}.
 

  }\vspace{-2em}}
 \maketitle

\begin{abstract}
 In this letter, the design of non-orthogonal multiple access (NOMA) beamforming is investigated in a spatial division multiple access (SDMA) legacy system. In particular, two popular beamforming strategies  in the NOMA literature, one to use existing SDMA beams and the other to form new beams, are adopted and compared.  The studies carried out in the letter  show that the two strategies realize different tradeoffs between system performance and complexity. For example, riding on existing beams  offers  a significant reduction in computational complexity, at the price of a slight performance loss.   Furthermore, this simple strategy  can realize the optimal performance when the users' channels are structured. 
\end{abstract}\vspace{-0.1em}

\vspace{-1.5em} 
 \section{Introduction}
 The design of beamforming is one of the most studied  topics in the research area of   non-orthogonal multiple access (NOMA) \cite{mojobabook}. In principle, there are two types of NOMA beamforming strategies. One is to encourage multiple users to share the same beamformer, whose rationale is to treat orthogonal spatial directions as a type of bandwidth resources and try to serve as many users as possible on  each spatial  direction \cite{Zhiguo_mimoconoma,7582424}. The other is to generate many non-orthogonal beamformers, where a single user is served on each of these non-orthogonal beams  \cite{7015589,7557079}. Both    principles have been shown superior  to   orthogonal multiple access (OMA),  and    applicable to various communication scenarios \cite{9371391,9389757,9531461}. 

This letter is to consider the NOMA beamforming design  in a  legacy system based on spatial division multiple access (SDMA), i.e., multiple primary users have been served via SDMA beamforming and additional secondary users are to be served via NOMA.  Following the     beamforming strategies    in the NOMA literature, two straightforward designs are to either use the existing SDMA beams or generate new ones dedicated for the secondary users. The optimization problems associated to the two strategies  are formulated and solved in the letter, where the potential loss of the optimality due to the use of semidefinite relaxation (SDR) is also analyzed.  The studies carried out in the letter  show that  riding on existing beams  offers  a significant reduction in computational complexity, at the price of a slight performance loss.   Furthermore, this simple strategy  can realize the optimal performance when the users' channels are structured, which opens up promising applications of reconfigurable intelligent surfaces (RISs) \cite{irs1}. 
 \vspace{-0.5em}
 \section{System Model}
Consider an   SDMA legacy system with  $K$ single-antenna primary users, denoted by ${\rm U}_k$, $1\leq k\leq K$, and one base station equipped with $N$ antennas, where the primary users are served by   using the   following zero-forcing (ZF)     precoder:
\begin{align}
\mathbf{P} = c \mathbf{H}\left(\mathbf{H}^H\mathbf{H}\right)^{-1},
\end{align}
where $\mathbf{H}=\begin{bmatrix}\mathbf{h}_1&\cdots &\mathbf{h}_K \end{bmatrix}$, $\mathbf{h}_k$ denotes ${\rm U}_k$'s channel vector, $c$ is the power normalization parameter, i.e., $c^2  = \frac{P_{\rm SDMA}}{tr \left(\mathbf{H} \left(\mathbf{H}^H\mathbf{H}\right)^{-2}\mathbf{H}^H\right)}$, $tr(\cdot)$ denotes the trace operation, and $P_{\rm SDMA}$ denotes the transmit power used for the primary users.

By using NOMA,    additional secondary users can be served together with  the primary users. For   illustration purposes, this letter focuses on  serving a single secondary user, denoted by ${\rm U}_0$.
Denote $\mathbf{x}$ by the signal sent by the base station, which can be written as follows: 
\begin{align}\label{signal}
\mathbf{x}=\mathbf{P}\mathbf{s}+\mathbf{w}s_0,
\end{align}
where $\mathbf{s}$ and $s_0$ are the signals sent to the primary users and ${\rm U}_0$, respectively, and $\mathbf{w}$ denotes ${\rm U}_0$'s beamforming vector. 

In this letter, it is assumed that the secondary user directly decodes its own signal by treating the primary users' signals as noise, because of its limited decoding capability.
 Depending on  the primary users'   SIC strategies,    ${\rm U}_0$'s data rate can be expressed  differently as follows.

\begin{itemize}
\item 
If none of the primary users carries out SIC, i.e., they   decode their  own signals   directly,  the secondary user's data rate is given by
\begin{align}\label{pr01}
R_0= \log\left(1 +\frac{|\mathbf{g}^H\mathbf{w}|^2}{|\mathbf{g}^H\mathbf{P}|^2+\sigma^2 }
 \right),
\end{align}
where $\mathbf{g}$ denotes the secondary user's channel vector and $\sigma^2$ denotes the noise power.
In order to strictly guarantee the primary users' QoS requirements,  the following   constraints are required: 
\begin{align}
\label{constraintx} 
  \log\left(1 +\frac{|\mathbf{h}_k^H\mathbf{p}_k|^2}{|\mathbf{h}_k^H\mathbf{w}|^2+\sigma^2 }
 \right)\geq R_i, 1\leq k \leq K,
\end{align}
where $\mathbf{p}_k$ denotes the $k$-th column of $\mathbf{P}$, and $R_i$ denotes ${\rm U}_i$'s target data rate.
Due to the use of ZF,    the interference  term, $\mathbf{h}_k^H\mathbf{p}_j$, $k\neq j$, can be discarded  in \eqref{constraintx} . 

\item Consider that  some primary users   decode the secondary user's signal first before decoding their owns. Denote the set containing these users by $\mathcal{S}$, and its complementary set by  $\mathcal{S}^c$.  The secondary user's data rate is given by
\begin{align}\label{eq05}
\tilde{R}_0=\min&\left\{  R_0, \log\left(1 +\frac{|\mathbf{h_i}^H\mathbf{w}|^2}{|\mathbf{h}_i^H\mathbf{P}|^2+\sigma^2 }
 \right), i \in \mathcal{S}  \right\}.
\end{align}
The constraint shown in \eqref{constraintx} is still needed for ${\rm U}_j$, $ j \in \mathcal{S}^c $, but not  for ${\rm U}_i$, $i \in \mathcal{S} $, since ${\rm U}_i$, $i \in \mathcal{S} $, is able to remove $s_0$ successfully and hence experiences the same performance as in conventional SDMA. 

\item If all the   primary users   decode the secondary user's signal first, the secondary user's data rate becomes
\begin{align}\label{rp6}
\bar{R}_0=\min&\left\{ R_0, \log\left(1 +\frac{|\mathbf{h_k}^H\mathbf{w}|^2}{|\mathbf{h}_k^H\mathbf{P}|^2+\sigma^2 }
 \right)  , 1\leq k \leq K\right\}.
\end{align}
The constraint shown in \eqref{constraintx}  is no longer needed, since the all the primary users experience the same performance as in conventional  SDMA.  
\end{itemize}

\section{Two  Strategies of Beamforming}
This letter is to investigate the use of two particular types of beamforming, termed Strategies I and II, as described in the following. 
\subsection{Strategy I - Riding on an Existing Beam}
The secondary user can simply use one of the existing beams, i.e., $\mathbf{w}=\sqrt{\alpha} \mathbf{p}_i$, where $ \mathbf{p}_i$ is ${\rm U}_i$'s beam,   $\alpha$ is to meet the power constraint, i.e., $\alpha\leq \frac{P_0}{|\mathbf{p}_i|^2}$, and $P_0$ denotes the transmit power budget for ${\rm U}_0$.

 \begin{itemize}
 \item 
  If ${\rm U}_i$ decodes its own signal   directly, by using 
  \eqref{pr01} and \eqref{constraintx}, the optimization problem of interest is given by
\begin{problem}\label{pb:1} 
  \begin{alignat}{2}
\hspace{-1em}\underset{\alpha}{\rm{max}} &     
\log\left(1 +\frac{\alpha|\mathbf{g}^H\mathbf{p}_i|^2}{|\mathbf{g}^H\mathbf{P}|^2+\sigma^2 }
 \right)    \\
\hspace{-1em}\rm{s.t.} &    \log\left(1 +\frac{|\mathbf{h}_i^H\mathbf{p}_i|^2}{\alpha|\mathbf{h}_i^H\mathbf{p}_i|^2+\sigma^2 }
 \right)\geq R_i ,   \alpha\leq \frac{P_0}{ |\mathbf{p}_i|^2}  
.
  \end{alignat}
\end{problem} 
Note that ${\rm U}_0$  does not cause interference to   ${\rm U}_j$, $j\neq i$, due to the use of ZF.  With some algebraic manipulations, the optimal solution is obtained as follows:
\begin{align}
\alpha^*_{I,i} = \min\left\{ \frac{P_0}{ |\mathbf{p}_i|^2} ,  \frac{\tau_i}{|\mathbf{h}_i^H\mathbf{p}_i|^2} \right\},
\end{align}
for $\mathbf{w}=\sqrt{\alpha} \mathbf{p}_i$, where $\tau_i = \max\left\{0, \frac{|\mathbf{h}_i^H\mathbf{p}_i|^2}{2^{R_i}-1}-\sigma^2\right\}$.
 
\item 
If ${\rm U}_i$ decodes ${\rm U}_0$'s signal first, by using \eqref{constraintx} and \eqref{rp6}, the optimization problem of interest is given by
\begin{problem}\label{pb:2} 
  \begin{align} 
\underset{\alpha}{\rm{max}} ~  &   
 \min\left\{\log\left(1 +\frac{\alpha |\mathbf{g}^H\mathbf{p}_i|^2}{|\mathbf{g}^H\mathbf{P}|^2+\sigma^2 }
 \right), \right.\\ \nonumber &\left.\log\left(1 +\frac{\alpha |\mathbf{h_i}^H\mathbf{p}_i|^2}{|\mathbf{h}_i^H\mathbf{P}|^2+\sigma^2 }
 \right)  \right\}
\rm{s.t.}  \quad    \alpha \leq \frac{P_0  }{ |\mathbf{p}_i|^2}\label{1st:3}
,
  \end{align}
\end{problem} 
whose solution is given by $ \alpha^*_{II,i} = \frac{P_0  }{ |\mathbf{p}_i|^2}$. 
 
\end{itemize}  
By comparing the data rates realized by $\alpha^*_{II,i} $ and $\alpha^*_{I,i}$,  the optimal solution for Strategy I can be found straightforwardly.

\subsection{Strategy II - Forming a New Beam}
Forming a new beam will be more complicated than using an existing SDMA beam, as shown in the following.
\begin{itemize}

\item If all the primary users decode their own signals   directly, by using 
  \eqref{pr01} and \eqref{constraintx}, the optimization problem of interest is given by
\begin{problem}\label{pb:3} 
  \begin{align}\nonumber
& \underset{\mathbf{w}}{\rm{max}}      
\log\left(1 +\frac{|\mathbf{g}^H\mathbf{w}|^2}{|\mathbf{g}^H\mathbf{P}|^2+\sigma^2 }
 \right)  \quad \rm{s.t.} \quad   |\mathbf{w}|^2\leq P_0,   \\ 
 &    \log\left(1 +\frac{|\mathbf{h}_k^H\mathbf{p}_k|^2}{|\mathbf{h}_k^H\mathbf{w}|^2+\sigma^2 }
 \right)\geq R_k, 1\leq k \leq K ,   
  \end{align}
\end{problem} 
which can be recast as the following equivalent form: 
 \begin{problem}\label{pb:4} 
  \begin{alignat}{2}
\hspace{-1em}\underset{\mathbf{w}}{\rm{max}} &\quad    
 |\mathbf{g}^H\mathbf{w}|^2  \label{1obj:4} \\
\hspace{-1em}\rm{s.t.} & \quad   |\mathbf{h}_k^H\mathbf{w}|^2\leq \tau_k, 1\leq k\leq K ,   |\mathbf{w}|^2\leq P_0  \label{1st:4}
.
  \end{alignat}
\end{problem}  
Problem \ref{pb:4}  is a nonconvex quadratically constrained quadratic program and can be solved by applying SDR \cite{5447068}, i.e., problem \ref{pb:4}  can be relaxed as follows:
 \begin{problem}\label{pb:5} 
  \begin{alignat}{2}
\hspace{-1em}\underset{\mathbf{W}}{\rm{max}} &\quad    
 tr\left(\mathbf{G}\mathbf{W} \right) \label{1obj:5} \\\hspace{-1em}
\rm{s.t.} & \quad   tr\left( \mathbf{H}_k\mathbf{W}\right) \leq \tau_k, 1\leq k \leq K ,  tr\left(\mathbf{W}\right)  \leq P_0  \label{1st:5}
,
  \end{alignat}
\end{problem}  
 where the rank-one constraint is omitted, $\mathbf{G}=\mathbf{g}\mathbf{g}^H$, and $\mathbf{H}_k=\mathbf{h}_k\mathbf{h}_k^H$. Problem \ref{pb:5} can be straightforwardly solved by applying optimization solvers. 
 
\item 
If ${\rm U}_i$, $i\in \mathcal{S}$, decodes the secondary user's signal first, by using \eqref{constraintx} and \eqref{eq05}, the optimization problem of interest is given by
\begin{problem}\label{pb:6} 
  \begin{align} 
\underset{\mathbf{w}}{\rm{max}}   &   
 \min\left\{\log\left(1 +\frac{|\mathbf{g}^H\mathbf{w}|^2}{|\mathbf{g}^H\mathbf{P}|^2+\sigma^2 }
 \right), \right.\\  \nonumber &\left. \log\left(1 +\frac{|\mathbf{h_i}^H\mathbf{w}|^2}{|\mathbf{h}_i^H\mathbf{P}|^2+\sigma^2 }
 \right)  , i\in\mathcal{S} \right\}
    \\\nonumber
\rm{s.t.} &    \log\left(1 +\frac{|\mathbf{h}_j^H\mathbf{p}_j|^2}{|\mathbf{h}_j^H\mathbf{w}|^2+\sigma^2 }
 \right)\geq R_j,  j\in\mathcal{S}^c ,     |\mathbf{w}|^2\leq P_0   
.
  \end{align}
\end{problem} 
Problem \ref{pb:6} can be recast equivalently as follows: 
 \begin{problem}\label{pb:7} 
  \begin{alignat}{2}
\underset{t, \mathbf{w}}{\rm{max}} &\quad    
t
  \label{1obj:7} \\
\rm{s.t.} & \quad  
\log\left(1 +\frac{|\mathbf{g}^H\mathbf{w}|^2}{|\mathbf{g}^H\mathbf{P}|^2+\sigma^2 }
 \right) \geq t\\&\quad
  \log\left(1 +\frac{|\mathbf{h_i}^H\mathbf{w}|^2}{|\mathbf{h}_i^H\mathbf{P}|^2+\sigma^2 }
 \right) \geq t,  i\in\mathcal{S}\\&\quad 
\log\left(1 +\frac{|\mathbf{h}_j^H\mathbf{p}_j|^2}{|\mathbf{h}_j^H\mathbf{w}|^2+\sigma^2 }
 \right)\geq R_j,  j\in\mathcal{S}^c,    \\&\quad   |\mathbf{w}|^2\leq P_0  
,
  \end{alignat}
\end{problem} 
which can be further simplified as follows:
 \begin{problem}\label{pb:8} 
  \begin{alignat}{2}
\underset{z, \mathbf{w}}{\rm{max}} &\quad    
z
  \label{1obj:8} \\
\rm{s.t.} & \quad  
 |\mathbf{g}^H\mathbf{w}|^2  
  \geq a_0 z\\&\quad
  |\mathbf{h_i}^H\mathbf{w}|^2   \geq a_i z,i\in\mathcal{S}\\&\quad 
 |\mathbf{h}_j^H\mathbf{w}|^2  \leq \tau_j ,  j\in\mathcal{S}^c,     |\mathbf{w}|^2\leq P_0  ,
  \end{alignat}
\end{problem} 
where $z=2^t-1$, $a_0=|\mathbf{g}^H\mathbf{P}|^2+\sigma^2$, and $a_i=|\mathbf{h}_i^H\mathbf{P}|^2+\sigma^2$. Again, problem \ref{pb:8} can be solved by applying SDR as shown previously. 

\item If all the primary users decode the secondary user's signal first before decoding their own, by using \eqref{constraintx} and \eqref{rp6}, the optimization problem of interest is given by
\begin{problem}\label{pb:9} 
  \begin{align}
\underset{\mathbf{w}}{\rm{max}} \quad    &
 \min\left\{\log\left(1 +\frac{|\mathbf{g}^H\mathbf{w}|^2}{|\mathbf{g}^H\mathbf{P}|^2+\sigma^2 }
 \right),\right. \\\nonumber &\left.  \log\left(1 +\frac{|\mathbf{h_i}^H\mathbf{w}|^2}{|\mathbf{h}_i^H\mathbf{P}|^2+\sigma^2 }
 \right) , 1\leq i \leq K \right\}
 \\
\rm{s.t.}  \quad    &   |\mathbf{w}|^2\leq P_0   
,
  \end{align}
\end{problem} 
which can be   recast equivalently as follows:
  \begin{problem}\label{pb:10} 
  \begin{alignat}{2}
\underset{z, \mathbf{w}}{\rm{max}} &\quad    
z
  \label{1obj:10} \\
\rm{s.t.} & \quad  
 |\mathbf{g}^H\mathbf{w}|^2  
  \geq a_0 z\\&\quad
  |\mathbf{h_i}^H\mathbf{w}|^2   \geq a_i z, 1\leq i \leq K, 
   |\mathbf{w}|^2\leq P_0   
.
  \end{alignat}
\end{problem}
Similar to problem \ref{pb:4},  problem \ref{pb:10} can be solved by applying SDR.  
\end{itemize}

After comparing the solutions obtained for problems \ref{pb:3}, \ref{pb:6} and \ref{pb:9}, a solution can be obtained for Strategy II accordingly. 
 
 \subsection{Optimality of the Obtained SDR Solutions}
 Because the SDR solutions have been obtained by removing the rank-one constraint, they are not guaranteed to be optimal. The optimality of the SDR solutions  for the two-user special case can be established as follows. Note that this conclusion can also be proved by applying Theorem 3.2 in \cite{5233822}, but the   steps shown  in the proof will be useful for the follow-up discussions about the general case with $K$ users.  
 
 \begin{proposition}\label{proposition1}
 For the two-user special case with random realizations of complex-valued channel coefficients,  the obtained solutions via SDR are optimal.  
 \end{proposition}
\begin{proof}
The proposition can be proved by simply showing that the obtained SDR solutions are always rank-one. Without loss of generality, problem \ref{pb:9} is focused, where the proofs for the other cases can be obtained straightforwardly. After applying SDR, problem \ref{pb:9} can be expressed  as follows:
  \begin{problem}\label{pb:11} 
  \begin{alignat}{2}
\underset{z, \mathbf{w}}{\rm{min}} &\quad    
-z
  \label{1obj:11} \\
\rm{s.t.} & \quad  
 tr\left( \mathbf{G} \mathbf{W}\right) 
  \geq a_0 z\\&\quad
 tr\left( \mathbf{H}_i\mathbf{W}\right)   \geq a_i z, 1\leq i \leq 2\\&\quad 
  tr\left(\mathbf{W}\right) \leq P_0 , 
  \mathbf{W}  \succeq 0  ,
  \end{alignat}
\end{problem} 
where the rank-one constraint  is ignored.  
The corresponding Lagrange can be written as follows \cite{8375979}:
\begin{align}
&L(z,\mathbf{W}, \lambda_1, \cdots,\lambda_4,\boldsymbol \lambda)\\\nonumber
=&-z + \lambda_1\left(a_0 z - tr\left( \mathbf{G} \mathbf{W}\right)  \right)+\sum^{3}_{i=2} \lambda_i\left(a_i z - tr\left( \mathbf{H}_i \mathbf{W}\right)  \right)
\\\nonumber
&+\lambda_4\left(  tr\left(  \mathbf{W}\right) -P_0 \right)-tr\left(\boldsymbol \lambda \mathbf{W}\right),
\end{align}
where $\lambda_i$ denotes the Lagrange multipliers. Because problem \ref{pb:11} is convex, the use of the Karush–Kuhn–Tucker (KKT) conditions is  applicable and leads  to the following: 
\begin{align}\label{kkt1}
 - \lambda_1   \mathbf{G}  -\sum^{3}_{i=2} \lambda_i  \mathbf{H}_i   
&+\lambda_4\mathbf{I}- \boldsymbol \lambda =0.
\end{align}
and
\begin{align}\label{kktxxx}
\boldsymbol \lambda \mathbf{W}=0.
\end{align}
In the remainder of the proof, we focus on the case with  $\mathbf{W}\neq 0$ to avoid the trivial all-zero solution.   

If $\boldsymbol \lambda=0$, the use of \eqref{kkt1} yields  the following:
 \begin{align}\label{eqlity1}
 \lambda_4\mathbf{I}= \lambda_1   \mathbf{G}  +\sum^{3}_{i=2} \lambda_i  \mathbf{H}_i  .
\end{align}
Note that $ \lambda_1   \mathbf{G}  +\sum^{3}_{i=2} \lambda_i  \mathbf{H}_i $ is a $2\times 2$ complex-valued matrix. For a fixed $\lambda_4$,   \eqref{eqlity1} represents a set of linear equations with $3$ unknown variables ($\lambda_1,\ldots, \lambda_3$) and $8$ equations. Because the channel coefficients are assumed to be random variables, no solution     exists for this overdetermined set of  linear equations, and hence  $\boldsymbol \lambda\neq 0$.
By using \eqref{kktxxx} and the fact that $\boldsymbol \lambda\neq 0$,   one can conclude that $\mathbf{W}$ is not full-rank. For the two-user case, if $\mathbf{W}$ is not full-rank,  the   rank of a feasible non-zero  $\mathbf{W}$ has to be one, which proves the proposition. 
\end{proof}

Note that there are a few situations, where   SDR solutions are not rank-one, as explained in the following. 
\subsubsection{The two-user case with $\mathbf{h}_1\perp \mathbf{h}_2$}
If $\mathbf{h}_1\perp \mathbf{h}_2$, by choosing $\lambda_1=0$, and $\lambda_{i-1}=\frac{1}{|\mathbf{h}_i|^2}$ for $2\leq i \leq 3$, $\sum^{3}_{i=2} \lambda_i  \mathbf{H}_i$ is a product of a unitary matrix and itself, and hence becomes an identity matrix. As a result, the equality in \eqref{eqlity1} becomes possible, i.e., $\boldsymbol \lambda =0$ becomes possible, which can lead to a situation that $\mathbf{W}$ is full-rank. 

\subsubsection{The two-user case with real-valued  channel coefficients} In this case, \eqref{eqlity1} can be viewed as a set of $3$ linear equations, one to ensure the diagonal elements of  $ \lambda_1   \mathbf{G}  +\sum^{3}_{i=2} \lambda_i  \mathbf{H}_i $ to be the same, the other two to ensure the off-diagonal elements to be zero. As a result, it is possible to find $\lambda_i$, $1\leq i \leq 3$, to satisfy \eqref{eqlity1}, i.e., $\boldsymbol \lambda =0$ becomes possible and  $\mathbf{W}$ is not necessarily rank-one. 

\subsubsection{The general case with $K>2$} For the general case, to establish the rank-one conclusion, it is not sufficient to just show $\mathbf{W}$ rank-deficient. The key step to establish the rank-one conclusion is to  rewrite  \eqref{kkt1} as follows:
\begin{align}\label{kkt2}
\boldsymbol \lambda = \lambda_{K+2}\mathbf{I}- \left(\lambda_1   \mathbf{G}  +\sum^{K+1}_{i=2} \lambda_i  \mathbf{H}_i   \right).
\end{align}
Assuming that $\left(\lambda_1   \mathbf{G}  +\sum^{K+1}_{i=2} \lambda_i  \mathbf{H}_i   \right)$ has two different   largest eigenvalues, denoted by $\lambda^*_1$ and $\lambda^*_2$, $\lambda^*_1>\lambda^*_2$, $\mathbf{W}$ can be proved to be rank-one, by using the following steps:
\begin{itemize}
\item If $\lambda^*_1<\lambda_{K+2}$, $\boldsymbol \lambda $ is full-rank, which is not possible due to the constraint $\boldsymbol \lambda \mathbf{W}=0$ shown in \eqref{kktxxx}. 

\item If $\lambda^*_1=\lambda_{K+2}$, and $\lambda^*_2< \lambda_{K+2}$, the dimension of the null space of $\boldsymbol \lambda $ is one, and hence $\mathbf{W}$ is rank-one, due to the constraint $\boldsymbol \lambda \mathbf{W}=0$.

\item If $\lambda^*_2\geq  \lambda_{K+2}$, $\boldsymbol \lambda $ has at least one  negative eigenvalue, which is not possible, since   $\boldsymbol \lambda $ is   positive semi-definite. 
\end{itemize} 

Unfortunately, our carried out simulation results indicate that   $\left(\lambda_1   \mathbf{G}  +\sum^{K+1}_{i=2} \lambda_i  \mathbf{H}_i   \right)$  can have two repeated   eigenvalues. When this situation happens,  the rank of $\mathbf{W}$ becomes two, and the use of Gaussian randomization procedure is needed.  
It is worth to point out that this rank-two situation is  observed for problem \ref{pb:9} only, but not for   problems \ref{pb:3} and \ref{pb:6}.. \vspace{-1em}
  \section{simulation}
  In this section, the computer simulation results are used to evaluate the performance of the two beamforming strategies. 
  
 A deterministic two-user case is first focused on in Fig. \ref{fig1}, where   Strategy II   yields the optimal solution\footnote{To avoid the sub-optimality issue discussed in Section III,  the users' real-valued channels are generated with small complex-valued  perturbations.   } , as indicated by Proposition \ref{proposition1}. The aim is to investigate   whether the use of Strategy I can also lead to the optimal performance. As can be seen from the figure, for the case where the primary users' channels are orthogonal to each other and the secondary user's channel is aligned with one primary user's, the use of Strategy I can yield the optimal performance. Fig. \ref{fig1} also shows that the optimal choice of $\mathbf{w}$ is closely related to the choices of $\mathbf{g}$ and   $P_0$ . For example, for $\theta=\frac{\pi}{4}$, with $P_0=27$ dBm,    $\mathbf{w}_{\rm Case 0}$ is optimal, whereas both $\mathbf{w}_{\rm Case 3-1}$ and  $\mathbf{w}_{\rm Case 3-2}$ are optimal if $P_0=30$ dBm.

 \begin{figure}[t]\centering \vspace{-3em}
    \epsfig{file=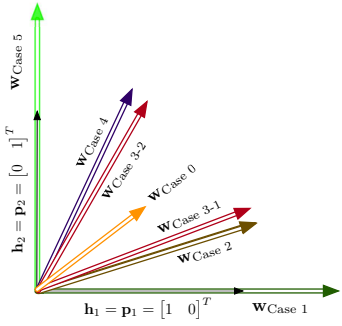, width=0.25\textwidth, clip=}\vspace{-0.5em}
\caption{ Illustration of  optimal beamforming, where   $\mathbf{g}= \begin{bmatrix} \sin\left(\theta\right) &\cos\left(\theta\right)\end{bmatrix}^T $ and $\theta$ is chosen as shown in Table \ref{table1}. For illustration purposes, large scale path loss is omitted, $P_{\rm SDMA}=30$ dBm, $N=K$, $R_i=1$ bit per channel use (BPCU), and $\sigma^2=-10$ dBm.  In particular,  $\mathbf{w}_{\text{Case 1}}=1.50\begin{bmatrix}1 &0\end{bmatrix}^T$, $\mathbf{w}_{\text{Case 2}}=\begin{bmatrix}1.18&0.32 \end{bmatrix}^T$, $\mathbf{w}_{\text{Case 3-1}}=\begin{bmatrix}1.13&0.47 	 \end{bmatrix}^T$,
$\mathbf{w}_{\text{Case 3-2}}=\begin{bmatrix}0.47 &1.13 	 \end{bmatrix}^T$, $\mathbf{w}_{\text{Case 4}}=\begin{bmatrix}0.32 &1.18\end{bmatrix}^T$, $\mathbf{w}_{\text{Case 5}}=1.50\begin{bmatrix}0 &1\end{bmatrix}^T$, $\mathbf{w}_{\text{Case 0}}=0.50\begin{bmatrix}1 &1\end{bmatrix}^T$. 
  \vspace{-2em}    }\label{fig1} 
\end{figure}

\begin{table}\vspace{-0em}
  \centering
  \caption{The Deterministic Cases Used to Generate Fig. \ref{fig1}   }\vspace{-1em}
\label{table1}
  \begin{tabular}{|c|l|c| l|}
   \hline
   & $\theta $	&$P_0$ in dBm	 &Adopted Beamforming   \\
    \hline
   Case 1& $\frac{\pi}{2}$	&$ 30$	 &Strategies  I and II  (SIC at ${\rm U}_1$)    \\\hline 
    Case 2&  $\frac{\pi}{3}$& $ 30$  &Strategy II   (SIC at ${\rm U}_1$)  \\\hline
        Case 3&  $\frac{\pi}{4}$& $ 30$  &Strategy II   (SIC at ${\rm U}_1$ or ${\rm U}_2$)   \\\hline
            Case 4&  $\frac{\pi}{6}$&$ 30$ &Strategy II   (SIC at ${\rm U}_2$ )  \\\hline
                            Case 5&  $0$& $ 30$  &Strategies I and  II   (SIC at ${\rm U}_2$)    \\\hline
                Case 0& $\frac{\pi}{4}$&$ 27$  &Strategy  II   (no SIC at ${\rm U}_1$ $\&$ ${\rm U}_2$)  \\\hline
  \end{tabular}\vspace{-1.5em}
\end{table}

In Figs. \ref{fig2} and \ref{fig3}, both path loss and small scale fading are considered. In particular, the $K$ primary users are randomly located in a square with edge $6$ m and the base station   located at its center, where the location of the secondary user is fixed at $(0,1)$ m.   The path loss exponent is set as $3$, the noise power is $\sigma^2=-94$ dBm, $N=K$, and $R_i=1$ bit per channel use (BPCU).   Fig. \ref{fig2} focuses on  the two-user case, and shows that   Strategy II outperforms Strategy I, at the price of high computational complexity.  In Fig. \ref{fig3}, a general multi-user case is considered.   Note that the implementation of Strategy II becomes prohibitively complicated for a large $K$ since there can be a huge number of possible $\mathcal{S}$ and  for each choice of $\mathcal{S}$ an SDR problem needs to be solved. Therefore, a simplified Strategy II is adopted by focusing on the following two cases, one to allow all the primary users to decode their own signals directly and the other to allow the primary user with the largest channel vector norm   to carry out SIC. As shown in Fig. \ref{fig3}, the performance gain of Strategy II   over  Strategy I is increased by increasing $K$; however, it is important to point out that the   complexity of SDR also grows by increasing $K$.

 \begin{figure}[t]\centering \vspace{-3em}
    \epsfig{file=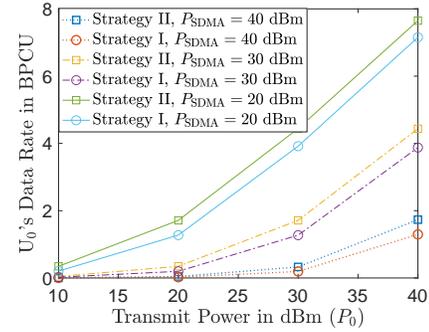, width=0.3\textwidth, clip=}\vspace{-0.5em}
\caption{ Performance comparison of the two strategies with different choices of $P_0$.
  \vspace{-1em}    }\label{fig2} 
\end{figure}

 \begin{figure}[t]\centering \vspace{-0.51em}
    \epsfig{file=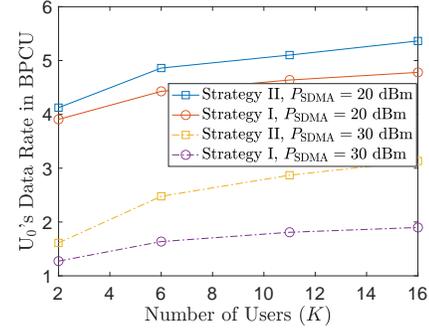, width=0.3\textwidth, clip=}\vspace{-0.5em}
\caption{ Performance comparison of the two strategies with different choices of $K$. $P_0=30$ dBm. 
  \vspace{-2em}    }\label{fig3} 
\end{figure}

\vspace{-1em}
\section{Conclusions}
 In this letter, the design of NOMA beamforming has been  investigated in an SDMA legacy system. In particular, two popular beamforming designs in the NOMA literature have been adopted and shown  to realize different tradeoffs between system performance and complexity.  
 \vspace{-1em}
 \section{Acknowledgements}
We thank Drs. Xiaofang Sun, Fang Fang, Jingjing Cui and Hong Xing for useful discussions related to SDR. 
 \vspace{-1em}
\bibliographystyle{IEEEtran}
\bibliography{IEEEfull,trasfer}

  \end{document}